\documentclass[10pt,a4paper]{article}
\usepackage[a4paper, total={6in, 8in}]{geometry}
\makeatother

\makeatother
\usepackage{graphicx}
\usepackage{amsthm}
\usepackage{amsmath}
\usepackage{tabu}
\usepackage{breqn}
\usepackage{verbatim}
\usepackage{caption}
\usepackage{subcaption}
\usepackage{esvect}
\usepackage{cite}
\usepackage{amssymb}
\usepackage{footnote}
\usepackage{color}
\usepackage{url}
\usepackage[makeroom]{cancel}
\usepackage{xcolor}
\usepackage{todonotes}
\newtheorem{theorem}{Theorem}[section]

\newtheorem{lemma}[theorem]{Lemma}
\newtheorem{proposition}[theorem]{Proposition}

\newtheorem{remark}{Remark}
\newtheorem*{problem statement}{Problem Statement}

\newtheorem{assumption}{Assumption}

\newcommand{\real}{\mathbb{R}}

\newcommand{\lie}{\mathcal{L}}

\newcommand{\Sset}{\mathcal{S}}
\newcommand{\ball}{\mathbb{B}}

\newcommand{\proj}{\mathrm{proj}}

\DeclareMathOperator*{\argmin}{arg\,min}

\title{On the relationship between control barrier functions and projected dynamical systems}
\author{Giannis Delimpaltadakis and W.P.M.H. Heemels\thanks{The authors are with the Control Systems Technology group, Mechanical Engineering, Eindhoven University of Technology. Emails:\texttt{\{i.delimpaltadakis, w.p.m.h.heemels\}@tue.nl}. \newline \indent This research received funding from the European Research Council (ERC) under the Advanced ERC grant agreement PROACTHIS, no. 101055384.}}
\date{}

\begin{document}

\maketitle
\begin{abstract}
    In this paper, we study the relationship between systems controlled via Control Barrier Function (CBF) approaches and a class of discontinuous dynamical systems, called \emph{Projected Dynamical Systems} (PDSs). In particular, under appropriate assumptions, we show that the vector field of CBF-controlled systems is a Krasovskii-like perturbation of the set-valued map of a differential inclusion, that abstracts PDSs. This result provides a novel perspective to analyze and design CBF-based controllers. Specifically, we show how, in certain cases, it can be employed for designing CBF-based controllers that, while imposing safety, preserve asymptotic stability and do not introduce undesired equilibria or limit cycles. Finally, we briefly discuss about how it enables continuous implementations of certain projection-based controllers, that are gaining increasing popularity.
\end{abstract}

\section{Introduction}
Control Barrier Functions (CBFs) have been introduced as a systematic methodology to impose safety on control systems \cite{ames2019control_review,ames2016control,reis2020control,tan2021undesired,wences2022compatibility,mestres2022optimization,yi2022complementarity_cbf}. Typically, given a nominal controller\footnote{Alternatively, a Control Lyapunov Function (CLF, see \cite{ames2019control_review}) is given, and the controller is designed by employing both the CBF and the CLF.}, which asymptotically stabilizes a given system, a modified version of it is designed, such that the system's trajectories never leave a safe set. The safe controller is specified as the solution of a Quadratic Program (QP), with its constraint, termed the CBF constraint, preventing leaving the safe set. The cost matrix in the QP's objective function is subject to design. A notorious problem in CBF-based controllers is that, for the sake of safety, undesired equilibria or limit cycles are often introduced, and thus stability is compromised \cite{reis2020control,tan2021undesired,wences2022compatibility,mestres2022optimization,yi2022complementarity_cbf}.

A seemingly unrelated topic is that of Projected Dynamical Systems (PDSs), a class of discontinuous constrained dynamical systems \cite{nagurney1995projected,heemels2000complementarity,brogliato2006equivalence,hauswirth2020anti,hauswirth2021projected,deenen2021projection,heemels2023existence,shi2022negative,lorenzetti2022pi,fu2023novel,heemels2020oblique}. They have been introduced to study dynamic aspects of constrained problems, such as variational inequalities \cite{nagurney1995projected} (e.g., traffic networks, oligopolies, energy markets etc.) and constrained optimization \cite{hauswirth2020anti,hauswirth2021projected}. Recently, they have, also, gained considerable interest for control purposes. They have been used to analyze a new class of hybrid integrators, called HIGS, which overcome limitations of linear controllers \cite{deenen2021projection,heemels2023existence,shi2022negative}. Moreover, they have been used in antiwindup control \cite{lorenzetti2022pi}, passivity-based control \cite{fu2023novel} and constrained observers \cite{heemels2020oblique}.

\subsection*{Contributions}
In this work, we investigate the relationship between CBFs and PDSs. The main result of this work can be summarised as: under suitable assumptions, we prove that the vector field of a system controlled via CBF methods is a Krasovskii-like perturbation (see \cite[Definition 6.27]{hybrid_book}) of the set-valued map of a differential inclusion (DI), that abstracts PDSs. In fact, we provide a \emph{quantitative} bound on that perturbation, that depends on the tunable parameter $a$ of the CBF constraint and is decreasing and vanishing with $a$.

This result provides a new perspective to analyze and design CBF-based controllers. To demonstrate this, we discuss how it can be used to design CBF-based controllers that preserve asymptotic stability of the nominal ones, while imposing safety and not introducing undesired equilibria or limit cycles. Specifically, in the simple scenario of complete control over the system dynamics ($\dot{x}=f(x)+u$), under the assumption of an incrementally stabilizing nominal controller and a convex safe set (although this can be relaxed; see Remark \ref{rem:simplifying_assumptions}), by picking $a$ sufficiently large and choosing an appropriate cost matrix for the QP, the CBF-based controller guarantees asymptotic stability and does not introduce undesired equilibria or limit cycles. This design procedure is showcased through a numerical example, which highlights that proper selection of the QP's cost matrix is an important design step, as poor choice thereof leads to loss of stability.

Finally, our result has other implications as well, which are not analyzed here, but are subject of future work. For example, as we mention in Section \ref{sec:future_work}, it enables continuous implementations of PDSs, in the form of CBF-controlled systems. This is of particular interest for projection-based controllers \cite{deenen2021projection,heemels2023existence,shi2022negative,lorenzetti2022pi,fu2023novel}, as continuous implementations of them can be beneficial, e.g., for additional robustness.

\subsection*{Related work}
In \cite{allibhoy2023control}, in the context of gradient flows (i.e., when the unconstrained dynamics is the gradient of an optimization problem's objective function), it is shown that, as $a\to\infty$, the CBF-controlled system tends to a PDS. Nonetheless, contrary to ours, this result refers only to the limit case, without providing a quantitative bound on the ``closeness" between the PDS and the CBF-controlled system. In fact, the result in \cite{allibhoy2023control} is partially\footnote{``Partially", as we consider the DI-representation of PDSs, whereas \cite{allibhoy2023control} considers the original discontinuous PDS vector field.} recovered by our result.

In \cite{hauswirth2020anti}, it is shown that a specific type of antiwindup control systems is a Krasovskii-like perturbation of PDSs. This result is subsequently used to derive algorithms for online feedback optimization with input constraints. Here we establish that a similar relationship exists between PDSs and CBF-controlled systems, which is, in general, a different class than the one of antiwindup control systems. Nonetheless, these two results suggest that the relationship between antiwindup control and CBFs is also of interest.

The problem of preserving asymptotic stability and eliminating undesired equilibria in CBF-based methods has received quite some interest lately \cite{reis2020control,tan2021undesired,wences2022compatibility,mestres2022optimization,yi2022complementarity_cbf}. Nonetheless, none of \cite{reis2020control,tan2021undesired,wences2022compatibility,mestres2022optimization,yi2022complementarity_cbf} approaches the problem by analyzing the PDS that is related to the CBF-controlled system. Thus, while there is no clear indication if our approach outperforms any of \cite{reis2020control,tan2021undesired,wences2022compatibility,mestres2022optimization,yi2022complementarity_cbf}, it provides a novel promising method and perspective to reason about and design CBF-controlled systems.

\section{Notation and Preliminaries}
Given a set $\Sset\subseteq\real^n$, denote its boundary by $\partial\Sset$ and its interior by $\mathrm{Int}(\Sset)$. Moreover, given any $x\in\real^n$, its Euclidean distance to $\Sset$ is $d(x,\Sset):= \min_{y\in\Sset}\|x-y\|$, where $\|\cdot\|$ denotes the Euclidean norm. Its projection to $\Sset$ is $\proj_{\Sset}(x) :=\argmin_{y\in\Sset}\|x-y\|$. The tangent cone to $\Sset$ at $x\in\Sset$, denoted by $T_\Sset(x)$, is the set of all vectors $w \in\real^n$ for which there exist sequences $\{x_i\}\in\Sset$ and $\{t_i\}$, $t_i>0$, with $x_i\to x$, $t_i\to 0$ and $i\to\infty$, such that $w=\lim_{i\to\infty}\frac{x_i-x}{t_i}$. Define the normal cone\footnote{As we work with sets satisfying \emph{constraint qualification conditions} (see, e.g., \cite{allibhoy2023control}), we do not distinguish between different kinds of normal cones.} of $\Sset$ at $x\in\Sset$ as $N_{\Sset}(x):=\{\eta\in\real^n\mid\eta^\top v\leq0, \text{ }\forall v\in T_\Sset(x)\}$. For more information on tangent and normal cones, see \cite{rockafellar2009variational}.

Denote the set of positive-definite symmetric matrices in $\real^{n\times n}$ by $\mathbb{S}_+^n$. Given $P\in\mathbb{S}_+^n$, denote its minimum and maximum eigenvalues by $\underline{\lambda}(P)$ and $\overline{\lambda}(P)$, respectively. Given $P\in\mathbb{S}_+^n$ and $x,y\in\real^n$, denote $\langle x\mid y\rangle_P=x^\top P y$ and $\|x\|_P=\sqrt{x^\top P x}$. Given $P\in\mathbb{S}_+^n$, a set-valued map $f:\real^n\rightrightarrows\real^n$ is called \emph{strongly $P$-monotone}, if for any $x,y\in\real^n$ and any $x'\in f(x)$, $y'\in f(y)$, one has $\langle x-y\mid x'-y'\rangle_P\geq \alpha\|x-y\|^2_P$, for some $\alpha>0$. Given a function $f:\real^n\to\real^{n\times m}$ and a differentiable function $h:\real^n\to\real$, denote $\lie_fh(x):=\nabla h^\top(x)\cdot f(x)$. A continuous function $\gamma:[0,\infty)\to\real$ is said to belong to $\mathcal{K}_\infty$, if $\gamma(0)=0$, $\gamma$ is strictly increasing, and $\lim_{a\to\infty}\gamma(a)=+\infty$.

\section{Control Barrier Functions and Projected Dynamical Systems}
Both CBF-based controlled systems and PDSs are systems constrained in some set $\Sset$. As commonly done in the literature \cite{ames2019control_review,ames2016control,reis2020control,tan2021undesired,wences2022compatibility,mestres2022optimization,yi2022complementarity_cbf}, we consider sets given as super-zero-level-sets of a \emph{control barrier function} $h:\real^n\to\real$: $$\Sset = \{x\in\real^n\mid h(x)\geq 0\}$$
We, also, make the following assumptions on $h$ and $\Sset$:
\begin{assumption}\label{assum:sset_and_h} $\Sset$ and $h$ satisfy the following:
\begin{enumerate}
    \item \label{assum:compactness} $\Sset$ is nonempty, compact, and $0\in\Sset$.
    \item \label{assum:h_C11} $h$ is continuously differentiable, $x\mapsto\nabla h(x)$ is locally Lipschitz, and its Lipschitz constant on $\Sset$ is $L_{\nabla h}$.
    \item \label{assum:regularity} For all $x\in\real^n$ such that $h(x)=0$, we have that $\nabla h(x)\neq0$.
    \item \label{assum:Kinf} There exists $\gamma\in\mathcal{K}_\infty$, such that $d(x,\partial\Sset)\leq\gamma(|h(x)|)$, for all $x\in\Sset$.
\end{enumerate}
\end{assumption}
Most of the above assumptions are standard in the literature of CBFs (that is, items \ref{assum:h_C11} and \ref{assum:regularity}, as well as item \ref{assum:compactness}, without the compactness assumption). Compactness of $\Sset$ is needed for several bounds in the proof of Theorem \ref{thm:main_result}. Extending Theorem \ref{thm:main_result} to non-compact sets is subject of future investigations. Moreover, item \ref{assum:Kinf} holds a-priori, if $h$ is real-analytic, by the \emph{Łojasiewicz inequality} (e.g., see \cite{tiep2012lojasiewicz}).

Consider the following ``nominal" (unconstrained) system:
\begin{equation}\label{eq:sys_nominal}
    \dot{x} = f(x)
\end{equation}
where $f:\real^n\to\real^n$. 
\begin{assumption}\label{assum:f_lipschitz}
$f$ is locally Lipschitz, and its Lipschitz constant on $\Sset$ is $L_{f}$. Moreover, the origin is the only equilibrium of \eqref{eq:sys_nominal} in $\Sset$.
\end{assumption}
Both CBF-based control methods and PDSs start from an unconstrained system \eqref{eq:sys_nominal} and end up with one that is constrained in $\Sset$. Given a tunable parameter $a>0$ and a $P\in\mathbb{S}_+^{n}$, consider
\begin{equation}\label{eq:sys_cbf}
    \dot{x} = f_{\mathrm{cbf},a}(x) := \left\{\begin{aligned}\argmin_\mu \quad&\|\mu-f(x)\|^2_P\\
    \mathrm{s.t.:} \quad &\lie_\mu h(x) + ah(x) \geq 0\end{aligned}\right. 
\end{equation}
Under the stated assumptions, the Quadratic Program (QP) in the above equation has a unique solution, for any $x$, which can be written in closed form as follows (see, e.g., \cite[Theorem 1]{wences2022compatibility}):
\begin{equation}\label{eq:f_cbf}
    f_{\mathrm{cbf},a}(x) = f(x) - \min\Big(0,\lie_f h(x) + ah(x)\Big)\frac{P^{-1}\nabla h(x)}{\|\nabla h(x)\|^2_{P^{-1}}}
\end{equation}
Essentially, the unconstrained vector-field is modified whenever $\lie_f h(x) + ah(x)\leq 0$, i.e., whenever the value of $h$ decreases along the trajectories of the nominal system \eqref{eq:sys_nominal} faster than a state-dependent tunable threshold $-ah(x)$. Vector fields such as \eqref{eq:sys_cbf} and \eqref{eq:f_cbf} are basically the closed-loop vector-fields in control systems of the form $\dot{x} = f(x)+u(x)$, where the controller $u(x)$ is designed via CBF-based methods. For more details, see the discussion in Section \ref{sec:design}. From standard CBF theory \cite{ames2019control_review}, it follows that trajectories of \eqref{eq:sys_cbf} starting in $\Sset$ stay in $\Sset$. Moreover, $f_{\mathrm{cbf},a}$ is locally Lipschitz, thereby implying existence and uniqueness of solutions of \eqref{eq:sys_cbf}.

Now, consider the discontinuous dynamical system
\begin{equation}\label{eq:sys_pds}
\dot{x} = f_{\mathrm{pds}}(x) := \left\{\begin{aligned}&f(x), \quad &x\in \mathrm{Int}(\Sset)\\&&\\
&\argmin_{\mu\in T_{\Sset}(x)} \|\mu-f(x)\|^2_P, \quad &x\in\partial\Sset
    \end{aligned}\right.
\end{equation}
Systems of the form \eqref{eq:sys_pds} are called \emph{Projected Dynamical Systems} (PDSs) \cite{nagurney1995projected}. In contrast to \eqref{eq:sys_cbf}, in PDSs, the modification of the vector field takes place \emph{only} at the boundary of $\Sset$. Specifically, at $x\in\partial\Sset$, a vector $\mu$ in the tangent cone of $\Sset$ is chosen (the one that minimizes $\|\mu-f(x)\|_P^2$), thus keeping trajectories inside $\Sset$. Under Assumptions \ref{assum:sset_and_h} and \ref{assum:f_lipschitz}, solutions of the PDS \eqref{eq:sys_pds} are equivalent to the solutions of the following Differential Inclusion (DI) (see \cite{hauswirth2021projected}):
\begin{equation}\label{eq:pds_di}
    \dot{x}\in F(x) := f(x)-P^{-1}N_{\Sset}(x), \quad x\in\Sset
\end{equation}
In what follows, we focus on \eqref{eq:pds_di}, as it is more amenable to analysis.

\subsection{Main result}
The main result of this work, Theorem \ref{thm:main_result} below, indicates that \eqref{eq:sys_cbf} and \eqref{eq:pds_di} are intimately related: $f_{\mathrm{cbf},a}$ belongs in a Krasovskii-like perturbation of the set-valued map $F$.
\begin{theorem}\label{thm:main_result}
Let Assumptions \ref{assum:sset_and_h} and \ref{assum:f_lipschitz} hold. For an arbitrarily small $\epsilon$, such that $0<\epsilon<\min_{z\in\partial\Sset}\|\nabla h(z)\|$, define
\begin{equation*}\label{eq:a_1}
        a_* = \frac{\max_{z\in\Sset}|\lie_f h(z)|}{\gamma^{-1}\Big(\frac{\min_{z\in\partial\Sset}\|\nabla h(z)\|- \epsilon}{L_{\nabla h}}\Big)}
\end{equation*}
Moreover, denote
\begin{align*}
    &M_1:= \min_{z\in\partial\Sset}\|\nabla h(z)\|\\
    &M_2:=\max_{z\in\partial\Sset}\|\nabla h(z)\|\\
    &M_3:= M_2 + L_{\nabla h}\gamma(\frac{1}{a_*}\max_{z\in\Sset}|\lie_f h(z)|)\\
    &L_1 := \frac{\overline{\lambda}(P)}{\underline{\lambda}(P)\epsilon^2}L_{\nabla h}\biggl[1 + \frac{M_2\overline{\lambda}(P)\Big(M_2+M_3\Big)}{\underline{\lambda}(P)M_1^2}\bigg]
\end{align*}
Finally, define
\begin{align*}
    \sigma(a,x) := \max\bigg\{&\gamma(\frac{1}{a}|\lie_f h(x)|),\\ &\Big(L_f+L_1|\lie_f h(x)|\Big)\gamma(\frac{1}{a}|\lie_f h(x)|)\bigg\}
\end{align*}
Then, for any $a\geq a_*$, it holds that for all $x\in\Sset$:
\begin{equation*}
    f_{\mathrm{cbf},a}(x) \in K_a(F(x))
\end{equation*}
where $K_a(F(x)):=F\Big((x+\sigma(a,x)\ball)\cap\Sset\Big) + \sigma(a,x)\ball$ and $\ball\subseteq\real^n$ is the closed unit ball.
\end{theorem}
\begin{proof}
    See Appendix.
\end{proof}
Observe that, due to Assumption \ref{assum:sset_and_h} item \ref{assum:regularity}, $M_1>0$, and, hence, there is always an $\epsilon$ such that $0<\epsilon<\min_{z\in\partial\Sset}\|\nabla h(z)\|$. Moreover, $\sigma$ is continuous on both arguments, strictly decreasing on $a$, and satisfies $\sigma(a,x)\geq0$, $\sigma(a,0)=0$ and $\lim_{a\to\infty}\sigma(a,x)=0$. As such, the larger $a$ is picked, the closer is the dynamics of \eqref{eq:sys_cbf} to the DI \eqref{eq:pds_di}.

\section{Preserving Stability in CBF-based Control}\label{sec:design}
As Theorem \ref{thm:main_result} provides \emph{quantitative} information, that depends on $a$, on how much of a perturbation $f_{\mathrm{cbf},a}$ is to $F$, and as that perturbation is vanishing as $a$ increases, by tuning $a$ sufficiently large, robust properties of $F$ may be transferred to $f_{\mathrm{cbf},a}$. In this section, we sketch how this perspective can be exploited to design the parameters of CBF-based controllers, such that they retain asymptotic stability of a given nominal controller, without introducing undesired equilibria or limit cycles, while guaranteeing safety. First, we provide a short discussion on the design methodology, and then we demonstrate it on a numerical example.

\subsection{Sketch of the design methodology}
We consider a control system $\dot{x}=f(x)+u$, for which a nominal control law\footnote{In the CBF literature, it is common that a nominal stabilizing (but, generally, unsafe) controller $u_0$ is given, and the CBF-based controller is derived as a modification of the nominal one (see, e.g., \cite{ames2019control_review,wences2022compatibility}).} $u_0:\Sset\to\real^n$ has been designed, without taking any safety considerations into account:  
\begin{equation}\label{eq:cbf_approach_nominal}
    \dot{x} = f_{0}(x) := f(x)+u_0(x)
\end{equation}
\begin{assumption}\label{assum:pds-based-control}
We impose the following assumptions:
\begin{itemize}
    \item $f_{0}$ is locally Lipschitz. Moreover, the origin is the only equilibrium of \eqref{eq:cbf_approach_nominal} in $\Sset$. Finally, $-f_{0}$ is strongly $G$-monotone, with $G\in\mathbb{S}^n_+$.
    \item $\Sset$ is convex.
\end{itemize}
\end{assumption}
In particular, strong $G$-monotonicity of $-f_{0}$ implies that \eqref{eq:cbf_approach_nominal} is \emph{quadratically incrementally stable}\footnote{Incremental stability is a form of stability where trajectories of the system converge to each other. For the technical definition, see, e.g., \cite{heemels2020oblique}.}; i.e., $u_0$ has been designed to incrementally stabilize the control system. As the origin is an equilibrium, incremental stability implies global asymptotic stability of the origin, with Lyapunov function $V(x)=x^\top G x$. Convexity of $\Sset$ is needed to make use of Proposition \ref{prop:maurice_deltagas} below. Nonetheless, as discussed in Remark \ref{rem:simplifying_assumptions}, this assumption can be relaxed, by extending Proposition \ref{prop:maurice_deltagas} to \emph{prox-regular sets} (see \cite{hauswirth2021projected} for a definition). 

Given the nominal controller $u_0$, CBF-based control methods design a safe controller $u_{\mathrm{cbf},a}:\Sset\to\real^n$ as follows: 
\begin{equation}\label{eq:ucbf_QP}
    u_{\mathrm{cbf},a}(x):= \left\{\begin{aligned}\argmin_v \quad&\|v-u_0(x)\|^2_P\\
    \mathrm{s.t.:} \quad &\lie_{f+v} h(x) + ah(x) \geq 0\end{aligned}\right. 
\end{equation}
where both $P\in\mathbb{S}^n_+$ and $a>0$ are subject to design. From CBF theory \cite{ames2019control_review} we know that $u_{\mathrm{cbf},a}$ renders $\Sset$ forward-invariant for the closed-loop (safety), for any $P$ and $a$. Here, we seek for a selection of $P$ and $a$ such that $u_{\mathrm{cbf},a}$, retains asymptotic stability, on top of its safety properties.

Towards this, we employ tools from PDS theory in combination with Theorem \ref{thm:main_result}. First, we recall the following:
\begin{proposition}[\hspace{-.1mm}{\cite[Theorem 2]{heemels2020oblique}} adapted]\label{prop:maurice_deltagas}
    Let Assumptions \ref{assum:sset_and_h} and \ref{assum:pds-based-control} hold. Then $-F_0(x):=-f_{0}(x)+G^{-1}N_{\Sset}(x)$ is strongly $G$-monotone.
\end{proposition}
The above proposition implies that, under its assumptions, any well-posed (in the sense of existence of solutions) system $\dot{x} = g(x)$, with $g(x)\in F_0(x) = f_{0}(x)-G^{-1}N_{\Sset}(x)$ is incrementally stable. In fact, it can be shown that, under suitable assumptions on $\sigma(a,x)$ (which, e.g., are satisfied in the numerical example of Section \ref{sec:numerical}), there exists some $a_{\mathrm{stable}}>0$ such that any system $\dot{x}=g(x)\in K_a(F_0(x))$ with $a\geq a_{\mathrm{stable}}$ is asymptotically stable, where $K_a(F_0(x))$ and $\sigma(a,x)$ are given by Theorem \ref{thm:main_result}.

Moreover, we observe that the closed-loop vector-field $f_{\mathrm{cl}} = f(x)+u_{\mathrm{cbf},a}(x)$ can be written as
\begin{equation}\label{eq:closed_loop_cbf}
    \dot{x} = f_{\mathrm{cl}}(x) := \left\{\begin{aligned}\argmin_\mu \quad&\|\mu-f_0(x)\|^2_P\\
    \mathrm{s.t.:} \quad &\lie_\mu h(x) + ah(x) \geq 0\end{aligned}\right. 
\end{equation}
which is like $f_{\mathrm{cbf},a}$ from \eqref{eq:f_cbf}, where $f$ is replaced by $f_0$.

The above suggest that the following selection for $P$ and $a$:
\begin{equation}\label{eq:P,a selection}
    P:=G, \quad a\geq\{a_*,a_{\mathrm{stable}}\}
\end{equation}
guarantees asymptotic stability for the CBF-controlled closed-loop $\dot{x} = f_{\mathrm{cl}}(x)$. That is because, by Theorem \ref{thm:main_result}, since $a\geq a_*$ and $P=G$, we have that $f_{\mathrm{cl}}(x)\in K_a(F_0(x))$, where $F_0(x) = f_{0}(x)-G^{-1}N_{\Sset}(x)$. Then, from what we mentioned above, since $a\geq a_{\mathrm{stable}}$, we know that $\dot{x}=f_{\mathrm{cl}}(x)\in K_a(F_0(x))$ is asymptotically stable.

Evidently, apart from the sufficiently large value of $a$, particular care is needed in choosing the cost matrix $P$ of the QP in \eqref{eq:ucbf_QP}: \emph{$P$ has to be equal to $G$, which comes from the Lyapunov function of the nominal system}. In fact, the numerical example in Section \ref{sec:numerical} shows that a wrong choice of $P$ in \eqref{eq:ucbf_QP} may result in stable undesired equilibria and loss of stability. That is because poor choice of $P$ might imply that the PDS $\dot{x}\in f_0(x) - P^{-1}N_{\Sset}(x)$ that is associated to the CBF-controlled system \eqref{eq:closed_loop_cbf} is not incrementally stable (see \cite[Example 1]{heemels2020oblique}), thus possibly implying instability of the CBF-dynamics, since PDS and CBF-dynamics are intimately related, as Theorem \ref{thm:main_result} suggests.

\begin{remark}\label{rem:new_perspectives}
    The proposed controller is similar to what is proposed in \cite{wences2022compatibility}, albeit in the simple scenario of $\dot{x}=f(x)+u$. Nonetheless, in \cite{wences2022compatibility}, no method on how to design the QP's cost matrix is given (except for a class of mechanical systems). This highlights how Theorem \ref{thm:main_result} and the connection between PDSs and CBFs provide new perspectives on how to design CBF-based controllers. 
\end{remark}
\begin{remark}\label{rem:simplifying_assumptions}
    To extend our results to general control-affine dynamics $\dot{x}=f(x)+g(x)u(x)$, that is typically considered in the CBF literature, we would have to consider PDSs with positive-semidefinite projection matrices $P$ or state-dependent ones $P(x)\in\mathbb{S}^n_+$. Moreover, the convexity assumption on $\Sset$ can be relaxed, by extending Proposition \ref{prop:maurice_deltagas} to e.g. prox-regular sets (for the definition of prox-regularity, see \cite{hauswirth2021projected}), by employing hypomonotonicity of their normal cones. These considerations are left for future work, as this work's main aim is to present Theorem \ref{thm:main_result} and provide a discussion on its possible implications. 
\end{remark}

 \subsection{Numerical Example}\label{sec:numerical}
 Consider the following 2-D control system, which has also been considered in \cite{mestres2022optimization} and \cite{yi2022complementarity_cbf}: 
 \begin{equation*}
     \dot{x} = x + u
 \end{equation*}
 where $x,u\in\real^2$. The nominal controller $u_0(x) = \begin{pmatrix}-2x_1-4x_2 &x_1-x_2\end{pmatrix}^\top$ incrementally stabilizes the nominal system $\dot{x}=f_{0}(x) = x+u_0(x)$ ($-f_0$ is strongly $G$-monotone), with a corresponding Lyapunov function $V(x)=x^\top G x$, where 
 \begin{equation}\label{eq:Pmatrix}
     G = \begin{pmatrix}
     0.625 &0.125\\ 0.125 &2.625
    \end{pmatrix}
\end{equation} 
Consider, also, the ellipsoidal set $\Sset = \{x\in\real^2\mid h(x)\geq 0\}$, where
$$h(x) = 9 - x^\top\begin{pmatrix}
    3 &2\\ 2 &2
\end{pmatrix} x$$
Observe that both $\Sset$ and $f_{0}$ satisfy Assumptions \ref{assum:sset_and_h} and \ref{assum:pds-based-control}. We seek to design a controller $u_{\mathrm{cbf},a}:\Sset\to\real^2$, such that the closed-loop system $\dot{x} = f_{\mathrm{cl}}(x) = x+u_{\mathrm{cbf},a}(x)$ is asymptotically stable, with $f_{\mathrm{cl}}(0)=0$, and the set $\Sset$ is forward invariant (i.e., safety is imposed).

 Figure \ref{fig:figure} depicts two trajectories, with the same initial condition: a) one that corresponds to the closed-loop system with the controller designed as in \eqref{eq:ucbf_QP} with the correct $P$-matrix (i.e., the one suggested by \eqref{eq:P,a selection}, which is given by \eqref{eq:Pmatrix}), and b) one that corresponds to a closed-loop system with a controller designed as in \eqref{eq:ucbf_QP}, but with a wrong $P$-matrix, namely $\begin{pmatrix}
     3 &0\\ 0 &1
 \end{pmatrix}$. In both cases, the parameter $a$ has been chosen as $a=1$. \emph{It is clear from the figure that the proposed controller (case a) safely stabilizes the origin, whereas the controller with the ``wrong" $P$-matrix (case b) introduces a stable undesired equilibrium at the boundary of $\Sset$ (at approximately $(-2.985,2.777)$). That is because in case (a) the associated PDS is incrementally stable, whereas in case (b) it is not.} This demonstrates the power of the insights that Theorem \ref{thm:main_result} provides, when designing CBF-based controllers.
\begin{figure}[h!]
     \centering
     \includegraphics[width=3.3in]{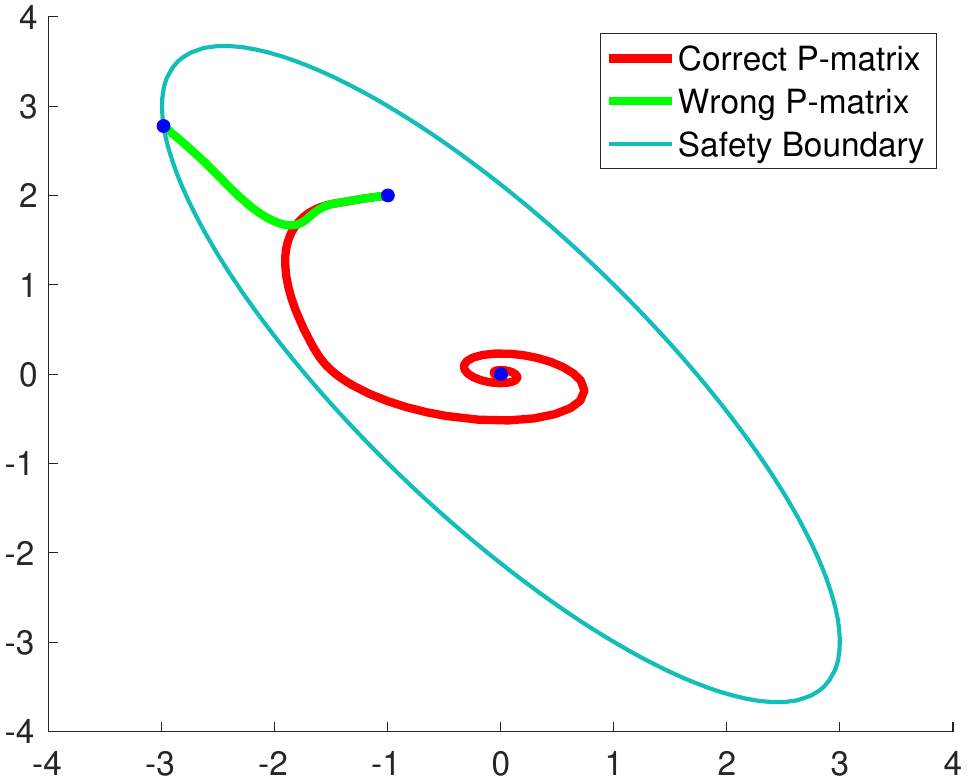}
     \caption{Trajectories of the closed-loop system for two different controllers: a) controller designed as in \eqref{eq:ucbf_QP} with the correct $P$-matrix (i.e., the one given by \eqref{eq:Pmatrix}), and b) controller designed as in \eqref{eq:ucbf_QP}, but with a wrong $P$-matrix. Initial condition is $(-1,2)$. The correct controller (case a) safely stabilizes the origin, whereas the wrong controller (case b) introduces a stable undesired equilibrium at $(-2.985,2.777)$.}
     \label{fig:figure}
 \end{figure}
 
 Remarkably, in the experiments it appeared that, irrespective of the value of $a$, the correct controller was able to stabilize the origin. This might be due to the linearity of the dynamics combined with the ellipsoidal safe set, as well as because $\sigma(a,x)$ in Theorem \ref{thm:main_result} has been derived in a conservative way (e.g., employing Lipschitz constants). Further investigations on this are subject of future work.

\section{Conclusion and Future Work}\label{sec:future_work}
We have proven that CBF-controlled systems are Krasovskii-like perturbations of PDSs, and provided a quantitative bound on that perturbation. This provides novel perspectives on analyzing and designing CBF-based controllers for safe stabilization. In the simple scenario of convex safe sets and complete control over the dynamics $\dot{x} = f(x) + u$, we have demonstrated how it can be used to design safe CBF-based controllers that preserve asymptotic stability of the origin and do not introduce undesired equilibria or limit cycles. Future work will focus on formalizing the design method of Section \ref{sec:design}, and extending it to more general dynamics and safe sets.

Finally, the implications of Theorem \ref{thm:main_result} extend beyond the design of CBF-based controllers. Following steps similar to \cite{hauswirth2020anti}, it can be shown that trajectories of the CBF-controlled system \eqref{eq:sys_cbf} uniformly converge to trajectories of the PDS \eqref{eq:sys_pds}. Thus, CBF-controlled systems may serve as continuous implementations/approximations of discontinuous PDSs. This is of particular interest for projection-based controllers \cite{deenen2021projection,heemels2023existence,shi2022negative,lorenzetti2022pi,fu2023novel}, as continuous implementations of such controllers might provide additional robustness.

\section*{Acknowledgements} The authors would like to thank Dr. Gabriel Gleizer for fruitful discussions on this work.

\section*{Appendix: Technical Lemmas and Proof of Theorem \ref{thm:main_result}}
In the following, we denote $U_{\mathrm{cbf},a}:=\{z\in\Sset\mid\lie_f h(z) + ah(z)\leq 0\}$. To prove Theorem \ref{thm:main_result}, we make use of the following Lemmas: 
\begin{lemma}\label{lem:neighbourhood}
    Given $a>0$, consider any $x\in U_{\mathrm{cbf},a}$ and any $y\in\proj_{\partial\Sset}(x)$. It holds that
    \begin{equation*}
        \|x-y\|\leq\gamma(\frac{1}{a}|\lie_f h(x)|)
    \end{equation*}
\end{lemma}
\begin{proof}[Proof of Lemma \ref{lem:neighbourhood}] It is proven by combining Assumption \ref{assum:sset_and_h} item \ref{assum:Kinf} with the fact that $0\leq h(x) \leq -\frac{1}{a}\lie_f h(x)$.
\end{proof}

\begin{lemma}\label{lem:nonzero_grad}
For any $a\geq a_*$, for all $x\in U_{\mathrm{cbf},a}$, we have   \begin{equation*}
    0<\epsilon\leq\|\nabla h(x)\|\leq M_3
\end{equation*}
\end{lemma}
\begin{proof}[Proof of Lemma \ref{lem:nonzero_grad}] It is proven by combining Assumption \ref{assum:sset_and_h} item \ref{assum:h_C11}, the triangle inequality and Lemma \ref{lem:neighbourhood}.
\end{proof}

\begin{lemma}\label{lem:normalized_grad_lipschitz}
    Given any $a\geq a_*$, it holds that, for any $x\in U_{\mathrm{cbf},a}$ and $y\in\mathrm{proj}_{\partial\Sset}(x)$
    \begin{equation*}
        \Big\|\frac{P^{-1}\nabla h(x)}{\|\nabla h(x)\|^2_{P^{-1}}} - \frac{P^{-1}\nabla h(y)}{\|\nabla h(y)\|^2_{P^{-1}}}\Big\|\leq L_1\|x-y\|
    \end{equation*}
    where
    \begin{equation*}
        L_1 := \frac{\overline{\lambda}(P)}{\underline{\lambda}(P)\epsilon^2}L_{\nabla h}\biggl[1 + \frac{M_2\overline{\lambda}(P)\Big(M_2+M_3\Big)}{\underline{\lambda}(P)M_1^2}\bigg]
    \end{equation*}
\end{lemma}
\begin{proof}[Proof of Lemma \ref{lem:normalized_grad_lipschitz}]
We have the following:
{
    \begin{align*}
    \Big\|\frac{P^{-1}\nabla h(x)}{\|\nabla h(x)\|^2_{P^{-1}}} - \frac{P^{-1}\nabla h(y)}{\|\nabla h(y)\|^2_{P^{-1}}}\Big\|=&\\
    \Big\|\frac{P^{-1}\nabla h(x)\|\nabla h(y)\|^2_{P^{-1}} - P^{-1}\nabla h(y)\|\nabla h(x)\|^2_{P^{-1}}}{\|\nabla h(x)\|^2_{P^{-1}}\|\nabla h(y)\|^2_{P^{-1}}}\Big\|
    \leq&\\ 
    \overline{\lambda}(P^{-1})\Bigg(\frac{\Big\|\nabla h(x)\cancel{\|\nabla h(y)\|^2_{P^{-1}}} - \nabla h(y)\cancel{\|\nabla h(y)\|^2_{P^{-1}}}\Big\|}{\|\nabla h(x)\|^2_{P^{-1}}\cancel{\|\nabla h(y)\|^2_{P^{-1}}}} +&\\
    +\frac{\Big\|\nabla h(y)\|\nabla h(y)\|^2_{P^{-1}} - \nabla h(y)\|\nabla h(x)\|^2_{P^{-1}}\Big\|}{\|\nabla h(x)\|^2_{P^{-1}}\|\nabla h(y)\|^2_{P^{-1}}}\Bigg)
    \leq&\\
    \overline{\lambda}(P^{-1})\Bigg(\frac{\|\nabla h(x)\ - \nabla h(y)\|}{\|\nabla h(x)\|^2_{P^{-1}}} +&\\
    \Big\|\nabla h(y)\Big\|\Big|\|\nabla h(y)\|_{P^{-1}} + \|\nabla h(x)\|_{P^{-1}}\Big|
    \cdot\frac{\Big|\|\nabla h(y)\|_{P^{-1}} - \|\nabla h(x)\|_{P^{-1}}\Big|}{\|\nabla h(x)\|^2_{P^{-1}}\|\nabla h(y)\|^2_{P^{-1}}}\Bigg)\leq&\\
    \overline{\lambda}(P^{-1})\bigg[\frac{\|\nabla h(x)-\nabla h(y)\|}{\underline{\lambda}(P^{-1})\epsilon^2}+
    M_2\overline{\lambda}(P^{-1/2})\Big(M_2+M_3\Big)
    \cdot\frac{\overline{\lambda}(P^{-1/2})\|\nabla h(x)-\nabla h(y)\|}{\underline{\lambda}^2(P^{-1})\epsilon^2M_1^2}    \bigg]\leq&\\
    \frac{\overline{\lambda}(P)}{\underline{\lambda}(P)\epsilon^2}L_{\nabla h}\biggl[1 + \frac{M_2\overline{\lambda}(P)\Big(M_2+M_3\Big)}{\underline{\lambda}(P)M_1^2}\bigg]\|x-y\|\quad&
\end{align*}}
where in the third inequality we used Lemma \ref{lem:nonzero_grad}.
\end{proof}

Let us proceed to the proof of Theorem \ref{thm:main_result}.
\begin{proof}[Proof of Theorem \ref{thm:main_result}] 
First, due to Assumption \ref{assum:sset_and_h} items \ref{assum:h_C11} and \ref{assum:regularity}, it holds that (see, e.g., \cite[Example 2.10]{hauswirth2021projected}):
\begin{equation*}
    N_{\Sset}(x)=\left\{\begin{aligned}
        &\{0\}, &\quad x\in\mathrm{Int}(\Sset)\\
        &\{\lambda\nabla h(x)\mid\lambda\leq0\}, &\quad x\in\partial\Sset
    \end{aligned}\right.
\end{equation*}
We distinguish the following cases:

\paragraph{Case 1} $x\in\Sset$ and $\lie_f h(x) + ah(x)>0$. Here we have $f_{\mathrm{cbf},a}(x)=f(x)\in F(x)$, and the result holds trivially.

\paragraph{Case 2} $x\in\Sset$ and $\lie_f h(x) + ah(x)\leq 0$. In this case, we can write
\begin{equation}
    f_{\mathrm{cbf},a}(x) = f(x) - \frac{\lie_f h(x) + ah(x)}{\|\nabla h(x)\|^2_{P^{-1}}}P^{-1}\nabla h(x)
\end{equation}
which is well-defined, as $\nabla h(x)\neq 0$, from Lemma \ref{lem:nonzero_grad}. Notice that, since $h(x)\geq0$ and $\lie_f h(x) + ah(x)\leq 0$, we have that $|\lie_f h(x)+ah(x)|\leq |\lie_f h(x)|$. 

Consider any $y\in\mathrm{proj}_{\partial\Sset}(x)$. Observe that $\eta = \Big(\lie_f h(x) + ah(x)\Big)\frac{\nabla h(y)}{\|\nabla h(y)\|^2_{P^{-1}}}\in N_{\Sset}(y)$, since $\frac{\lie_f h(x) + ah(x)}{\|\nabla h(y)\|^2_{P^{-1}}}\leq 0$. Also, $\nabla h(y) \neq 0$, due to Assumption \ref{assum:sset_and_h} item \ref{assum:h_C11}. We have the following:
\begin{align*}
    \Big\|f(y) - P^{-1}\eta - f(x) + \frac{\lie_f h(x) + ah(x)}{\|\nabla h(x)\|^2_{P^{-1}}}P^{-1}\nabla h(x)\Big\|\leq&\\
    \|f(y)-f(x)\| +|\lie_f h(x) + ah(x)|\Big\|\frac{P^{-1}\nabla h(y)}{\|\nabla h(y)\|^2_{P^{-1}}}-\frac{P^{-1}\nabla h(x)}{\|\nabla h(x)\|^2_{P^{-1}}}\Big\|\leq&\\
    L_f\|x-y\|+|\lie_f h(x) + ah(x)|L_1\|x-y\|\leq&\\
    L_f\|x-y\|+|\lie_f h(x)|L_1\|x-y\|\quad\leq&\\
    \underbrace{\Big(L_f+L_1|\lie_f h(x)|\Big)\gamma(\frac{1}{a}|\lie_f h(x)|)}_{\sigma_1(a,x)}\quad
\end{align*}
where in the second inequality we used Assumption \ref{assum:f_lipschitz}, in the third inequality we used Lemma \ref{lem:normalized_grad_lipschitz}, in the fourth inequality we used that $|\lie_f h(x)+ah(x)|\leq|\lie_f h(x)|$, and in the fifth inequality we used Lemma \ref{lem:neighbourhood}. Finally, from the above we get:
\allowdisplaybreaks{\begin{align*}
    f_{\mathrm{cbf},a}(x) =&\\ 
    f(x) - \frac{\lie_f h(x) + ah(x)}{\|\nabla h(x)\|^2_{P^{-1}}}P^{-1}\nabla h(x) \in&\\
    f(y) - P^{-1}\eta+\sigma_1(a,x)\ball\subseteq&\\
    f(y) - P^{-1}N_{\Sset}(y)+\sigma_1(a,x)\ball\subseteq&\\
    f\Big((x+\gamma(\frac{1}{a}|\lie_f h(x)|)\ball)\cap\Sset\Big) - P^{-1}N_{\Sset}\Big((x+\gamma(\frac{1}{a}|\lie_f h(x)|)\ball)\cap\Sset\Big)+\sigma_1(a,x)\ball\subseteq&\\
    f\Big((x+\sigma(a,x)\ball)\cap\Sset\Big) - P^{-1}N_{\Sset}\Big((x+\sigma(a,x)\ball)\cap\Sset\Big)+\sigma(a,x)\ball=&\\
    F\Big((x+\sigma(a,x)\ball)\cap\Sset\Big) + \sigma(a,x)\ball\quad&
\end{align*}}
\end{proof}
\bibliography{mybib.bib}
\bibliographystyle{IEEEtran}
\end{document}